\documentclass[a4paper,onecolumn,accepted=2024-07-10]{quantumarticle}
\pdfoutput=1
\usepackage[utf8]{inputenc}
\usepackage[T1]{fontenc}
\usepackage[english]{babel}
\usepackage{amsmath,amssymb,amsthm}
\usepackage{hyperref}
\usepackage{url}
\usepackage{subcaption}
\usepackage{mathtools}
\usepackage{booktabs}
\usepackage{csquotes}
\usepackage[block=space,
backend=biber,
bibstyle=alphabetic,
citestyle=alphabetic,
maxbibnames=10,
]{biblatex}
\renewbibmacro{in:}{}
\bibsetup{
  \setcounter{abbrvpenalty}{9000}
  \setcounter{highnamepenalty}{9000}
  \setcounter{lownamepenalty}{9000}
  \setcounter{biburlnumpenalty}{9000}
}

\DeclareMathSymbol{\shortminus}{\mathbin}{AMSa}{"39}

\usepackage{tikz}
\usetikzlibrary{automata,arrows.meta,positioning,quotes}
\usetikzlibrary{calc}
\usetikzlibrary{shapes.symbols}
\usetikzlibrary{fit}
\usetikzlibrary{shapes.geometric}
\usepackage{siunitx}
\usepackage{physics}

\usepackage{subcaption} 

\AtEveryBibitem{\clearfield{month}}
\AtEveryBibitem{\clearfield{issue}}
\AtEveryBibitem{\clearfield{primaryclass}}
\AtEveryBibitem{\clearfield{archivePrefix}}

\theoremstyle{plain}
\newtheorem{thm}{Theorem}
\newtheorem*{thm*}{Theorem}

\newtheorem{fakt}{Fact}

\newtheorem{assu}[thm]{Assumption}

\theoremstyle{definition}

\theoremstyle{remark}

\begin{document}
\definecolor{black}{rgb}{0,0,0}

\newcommand{\B}{{\mathsf{B}}}
\newcommand{\T}{{\mathsf{T}}}


\newcommand{\beq}[0]{\begin{equation}}
\newcommand{\eeq}[0]{\end{equation}}
\newcommand{\srn}[2]{\big\langle #1 \big\rangle_{#2}}
\newcommand{\outerp}[2]{\ket{#1}\! \bra{#2}}

\def\ra{\rangle}
\def\la{\langle}

\newcommand{\one}{\leavevmode\hbox{\small1\normalsize\kern-.33em1}}
\def\tr{\mbox{tr}}
\def\compl{\mathbb{C}}

\font\Bbb =msbm10 \font\eufm =eufm10
\def\Real{{\hbox{\Bbb R}}} \def\C{{\hbox {\Bbb C}}}
\def\LRA{\mathop{-\!\!\!-\!\!\!\longrightarrow}\nolimits}
\def\LRAA{\mathop{\!\longrightarrow}\nolimits}

\def\be{\begin{equation}}
\def\ee{\end{equation}}
\def\ben{\begin{eqnarray}}
\def\een{\end{eqnarray}}
\def\eea{\end{array}}
\def\bea{\begin{array}}
\newcommand{\ot}[0]{\otimes}
\newcommand{\bei}{\begin{itemize}}
\newcommand{\eei}{\end{itemize}}
\newcommand{\wektor}[1]{\boldsymbol{#1}}
\newcommand{\proj}[1]{\ket{#1}\!\!\bra{#1}}
\newcommand{\Ke}[1]{\big|#1\big\rangle}
\newcommand{\Br}[1]{\big< #1\big|}
\newcommand{\mo}[1]{\left|#1\right|}
\newcommand{\nms}{\negmedspace}
\newcommand{\nmss}{\negmedspace\negmedspace}
\newcommand{\nmsss}{\negmedspace\negmedspace\negmedspace}
\newcommand{\norsl}[1]{\left\|#1\right\|_{\mathrm{1}}}
\newcommand{\norhs}[1]{\left\|#1\right\|_{\mathrm{2}}}
\newcommand{\ke}[1]{|#1\rangle}
\newcommand{\br}[1]{\langle #1|}
\newcommand{\kett}[1]{\ket{#1}}
\newcommand{\braa}[1]{\bra{#1}}
\def\ra{\rangle}
\def\la{\langle}
\def\blacksquare{\vrule height 4pt width 3pt depth2pt}
\def\dcal{{\cal D}}
\def\pcal{{\cal P}}
\def\hcal{{\cal H}}
\def\A{{\cal A}}
\def\B{{\cal B}}
\def\E{{\cal E}}
\def\trace{\mbox{Tr}}
\newtheorem{theo}{Theorem}
\newtheorem{conj}[theo]{Conjecture}

\newcommand{\R}{\mathbb{R}}
\newcommand{\I}{\mathbbm{1}}
\newcommand{\w}{\omega}
\newcommand{\p}{\Vec{p}}
\newcommand{\dl}[1]{\left|\left|#1\right|\right|}

\renewcommand{\emph}[1]{\textbf{#1}}
\newcommand{\emphalt}[1]{\textit{#1}}

\newcommand{\eg}{{\it{e.g.~}}}
\newcommand{\ie}{{\it{i.e.~}}}
\newcommand{\etal}{{\it{et al.~}}}


\title{Network quantum steering enables randomness certification
without seed randomness}
\author{Shubhayan Sarkar}
\affiliation{Laboratoire d’Information Quantique, Université libre de Bruxelles (ULB), Av. F. D. Roosevelt 50, 1050 Bruxelles, Belgium}
\orcid{0000-0001-5833-4466}

\begin{abstract}
   Quantum networks with multiple sources allow the observation of quantum nonlocality without inputs. Consequently, the incompatibility of measurements is not a necessity for observing quantum nonlocality when one has access to multiple quantum sources. Here we investigate the minimal scenario without inputs where one can observe any form of quantum nonlocality. We show that even two parties with two sources that might be classically correlated can witness a form of quantum nonlocality, in particular quantum steering, in networks without inputs if one of the parties is trusted, that is, performs a fixed known measurement. We term this effect as swap-steering. The scenario presented in this work is minimal to observe such an effect. Consequently, a scenario exists where one can observe quantum steering but not Bell non-locality. We further construct a linear witness to observe swap-steering. Interestingly, this witness enables self-testing of the quantum states generated by the sources and the local measurement of the untrusted party. This in turn allows certifying two bits of randomness that can be obtained from the measurement outcomes of the untrusted device without the requirement of initially feeding the device with randomness.
\end{abstract}

\section{Introduction} 
Quantum nonlocality is one of the most remarkable features of quantum mechanics that defy our classical intuitions about the world. It refers to the property of quantum particles to exhibit correlations that seem to occur instantaneously even when they are separated by large distances. This quantum property was first conceptualized in the celebrated work of Einstein, Podolsky and Rosen \cite{EPR}. Based on it, Bell in 1964 \cite{Bell, Bell66} proposed a theoretical test, known as Bell's inequality, that could distinguish between classical and quantum correlations. It was then experimentally verified \cite{Bellexp1, Bellexp2, Bellexp3, Bellexp4} and is now recognized as a fundamental aspect of quantum mechanics. The implications of quantum nonlocality are far-reaching, with potential applications in fields such as cryptography, quantum teleportation, quantum communication, and quantum computing (refer to \cite{NonlocalityReview} for a review).

Another form of quantum nonlocality, known as quantum steering, allows for one observer to remotely influence the state of another observer's quantum system, even if the two observers are separated by large distances. Quantum steering was first conceptualized by Schrodinger \cite{Schrod} and was then rigorously introduced in \cite{Wiseman}. The major difference between the scenarios to observe Bell nonlocality and quantum steering is that one of the parties is assumed to be trusted in the latter one, that is, known to perform fixed measurements.

To observe quantum nonlocality or quantum steering, any party involved in the experiment must have at least two inputs as incompatible measurements are necessary to witness any of these phenomena. Interestingly, quantum networks allow for witnessing such non-classical features without the requirement of incompatibility of measurements. The framework to witness quantum nonlocality in networks was introduced in \cite{pironio1,Pironio22, Fritz}. However, it was first noted in \cite{Pironio22} and then in \cite{Fritz} that considering independent sources shared between non-communicating parties allows one to observe quantum nonlocality with a single fixed measurement for every party. Recently, the authors in \cite{renou1, renou2, renou3, supic4} explore this phenomenon to construct scenarios where one can observe genuine quantum network nonlocality.

One of the intriguing problems in this regard concerns the minimal scenario in which any form of quantum nonlocality can be observed without any inputs. It was shown in \cite{renou1}, that genuine network nonlocality can be observed without inputs if there are three parties with three independent sources. Inspired by entanglement swapping \cite{swap}, we show here that if one of the parties is assumed to be trusted then one can observe a form of quantum nonlocality, which we term as swap-steering, using only two parties and two sources. Unlike most of the considered quantum network scenarios where one assumes independence of the sources [see nevertheless Ref. \cite{ivan2,Sarkar_2024}], we relax this assumption and allow the sources to be classically correlated. Moreover, the swap-steering scenario is the minimal scenario where one can observe a form of quantum nonlocality without inputs. Further on, there is a lack of witnesses when observing quantum nonlocality without inputs in networks. This restricts the possibility of testing these phenomena at the operational level. Interestingly, we find a linear witness to observe swap-steering thus, making our notion of nonlocality experimentally testable. We further identify some states that are unsteerable in the standard quantum steering scenario are swap-steerable. In particular, any entangled two-qubit Werner state 
is swap-steerable, which can be interpreted as an entanglement-assisted activation of quantum steering.

As an application of our work, we utilize the above result for one-sided device-independent (DI) certification where one can completely characterize the states generated by the sources and the untrusted measurements up to some degrees of freedom. Using the outcomes of the certified measurement, one can then generate genuine randomness even when an intruder might have access to them. This is extremely important for any cryptographic scheme as the security of these schemes relies on access to random number generators. 
Moreover, any of the known schemes for DI certification of states, measurements or randomness requires access to seed randomness, that is, the measurement devices whose outcomes will be used to generate random numbers, have inputs that have to be chosen randomly in order for the protocol to be secure [for instance see Refs. \cite{di4, random0, rand1, rand2, rand3, Fehr, APP13, Armin1, chainedBell, sarkar, sarkar5}]. Furthermore, DI certification of quantum states and measurements in quantum networks was recently explored in Refs. \cite{Marco, NLWEsupic, JW2, Allst1, supic4, sekatski, sarkar2023}. However, all of these certification schemes require at least two inputs for most of the measurement devices. A partial certification scheme was proposed in \cite{sekatski} that utilizes the genuine network nonlocality without inputs in a triangle network \cite{renou1}. However, using the proposed scheme \cite{sekatski}, one can only conclude that the sources need to prepare entangled states with at least $2.5\ \%$ of entanglement of formation and one can securely extract randomness of .04 bits. We utilize the maximal violation of the proposed swap-steering inequality for self-testing the singlet state along with the Bell basis which is then used for generating secured randomness of two bits without the requirement to initially feed the devices with random numbers. This is the first instance where the exact certification of quantum states, measurements, and randomness could be achieved without inputs.

\section{The scenario} 

In this work, we consider the simplest scenario consisting of two parties namely, Alice and Bob in two different labs far away from each other. Both of them receive two subsystems from two different sources $S_1, S_2$ that might be classically correlated to each other. Now they perform a single four-outcome measurement on their respective subsystems where the outcomes are denoted as $a,b=0,1,2,3$ respectively for Alice and Bob. Alice is trusted here implying that the measurement performed by her on her subsystems is known (see Fig. \ref{fig1}). We consider here that she performs the measurement corresponding to the Bell basis given by $ M_{A}=\{\proj{\phi_{+}},\proj{\phi_{-}},\proj{\psi_{+}},\proj{\psi_{-}}\}_{A_1A_2} $ where

\begin{eqnarray}\label{Amea1}
    \ket{\phi_{\pm}}_{A_1A_2}&=&\frac{1}{\sqrt{2}}\left(\ket{0}_{A_1}\ket{0}_{A_2}\pm\ket{1}_{A_1}\ket{1}_{A_2}\right)\nonumber\\
    \ket{\psi_{\pm}}_{A_1A_2}&=&\frac{1}{\sqrt{2}}\left(\ket{0}_{A_1}\ket{1}_{A_2}\pm\ket{1}_{A_1}\ket{0}_{A_2}\right).
\end{eqnarray}
Here $A_1/A_2, B_1/B_2$ denote the two different subsystems of Alice and Bob respectively. Notice that in the particular case when the sources generate the singlet state, the above scenario is equivalent to entanglement swapping.

Now, Alice and Bob repeat the experiment enough times to construct the joint probability distribution (correlations) $\vec{p}=\{p(a,b)\}$ where $p(a,b)$ denotes the probability of obtaining outcome $a,b$ with Alice and Bob respectively. These probabilities can be computed in quantum theory as
\begin{equation}
p(a,b)=\sum_jp_j\Tr\left[(M^a\otimes N^b)\rho_{A_1B_1}^j\otimes\rho_{A_2B_2}^j\right]
\end{equation}
where $M^a,N^b$ denote the measurement elements of Alice and Bob which are positive and $\sum_aM^a=\sum_bN^b=1$ and $\sum_jp_j=1$. 
It is important to recall here that Alice and Bob can not communicate with each other during the experiment. 

\begin{figure}[t]
\begin{center}
\includegraphics[width=12cm]{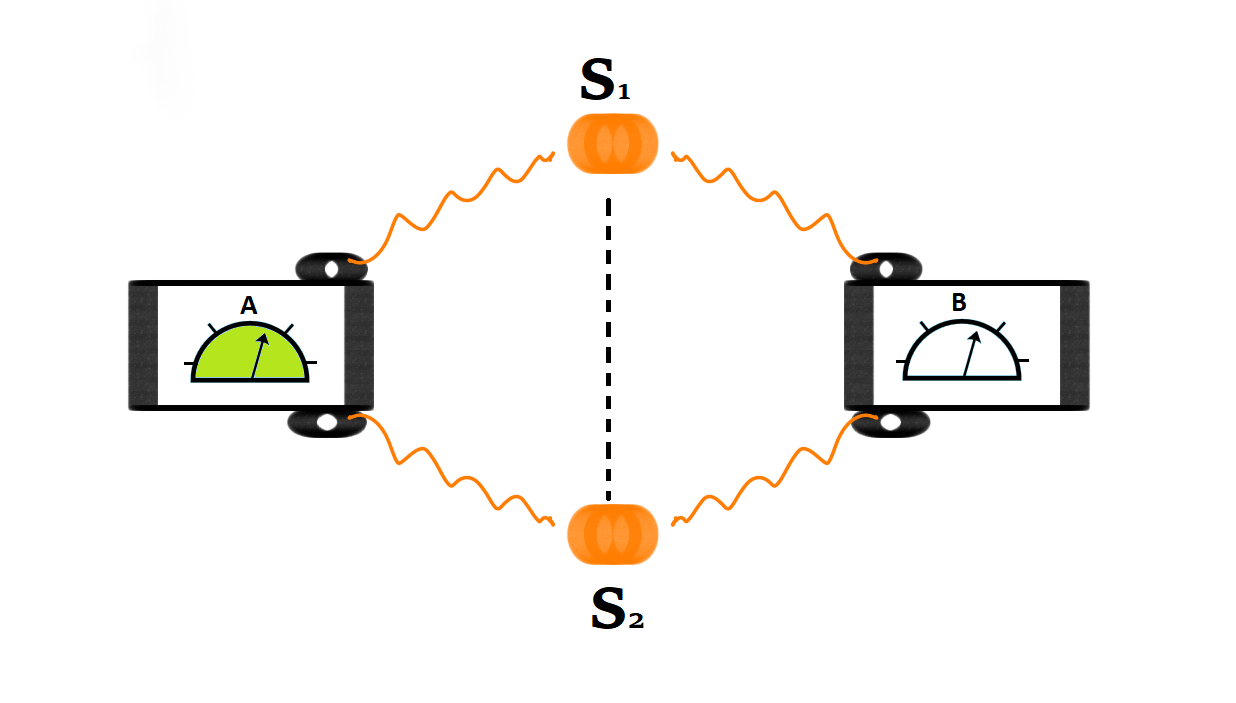}
    \caption{Swap-steering scenario. Alice and Bob are spatially separated and each of them receives two subsystems from the sources $S_1, S_2$. On the received subsystem they perform a single four-outcome measurement. Alice is trusted here, meaning that she is known to perform the Bell-basis measurement. They are not allowed to communicate during the experiment, however, the sources might classically communicate with each other. Once it is complete, they construct the joint probability distribution $\{p(a,b)\}$.}
    \label{fig1}
    \end{center}
\end{figure}

\section{Swap-steering}
Suppose that there are some variables $\lambda_i$ that are being sent by the sources $S_i$ as depicted in Fig. \ref{fig2}. 
Further on, as Alice is known to perform quantum measurements, the variable she receives is some quantum state $\rho_{\lambda_1,\lambda_2}$, however, there is no such restriction on Bob.
Let us now state the two assumptions, namely outcome-independence and separable quantum sources, that must be satisfied if Bob is classical, or equivalently if the correlations are not swap-steerable from Bob to Alice.

\begin{assu}[Outcome-independence] The outcomes of two parties are independent of each other if one has access to the hidden variables $\lambda_i$.
\end{assu}

In the scenario considered in this work, Bob's outcome $b$ being independent of Alice's outcome $a$ means that for any $a,b,\lambda_1,\lambda_2$,
\begin{eqnarray}
p(b|\lambda_1,\lambda_2,a)=p(b|\lambda_1,\lambda_2).
\end{eqnarray}
This is a weaker definition of locality when compared to Bell's assumption of locality, or the notion of locality in the standard quantum steering scenario.  
\begin{assu}[Separate quantum sources]\label{ass2} Two sources $S_i$ $(i=1,2)$ generating a joint quantum state $\rho_{\lambda_1,\lambda_2}$ are separate if the state $\rho_{\lambda_1,\lambda_2}$ is separable for any $\lambda_1,\lambda_2$. 
\end{assu}

Notice that the in above assumption \ref{ass2}, we impose on the sources is weaker when compared to independent quantum sources. As a matter of fact, the above assumption allows the sources to communicate classically with each other or equivalently the sources might generate classically correlated states. 
Now, given two sources $S_i$ for $i=1,2$ that generate some (for now hidden) states $\lambda_i$, we can always express the probability $p(a,b)$ as
\begin{eqnarray}
    p(a,b)=\sum_{\lambda_1,\lambda_2} p(\lambda_1,\lambda_2) p(a,b|\lambda_1,\lambda_2).
\end{eqnarray}
Using Bayes rule and the fact that Alice is known to be performing quantum measurements, we can express the above expression as
\begin{equation}
     p(a,b)=\sum_{\lambda_1,\lambda_2} p(\lambda_1,\lambda_2) p(a|\rho_{\lambda_1,\lambda_2})p(b|\lambda_1,\lambda_2,a).
\end{equation}
Assuming outcome-independence, we arrive at
\begin{eqnarray}
     p(a,b)=\sum_{\lambda_1,\lambda_2} p(\lambda_1,\lambda_2) p(a|\rho_{\lambda_1,\lambda_2})p(b|\lambda_1,\lambda_2).
\end{eqnarray}
\begin{figure*}[t!]%
    \centering
    {{\includegraphics[width=7cm]{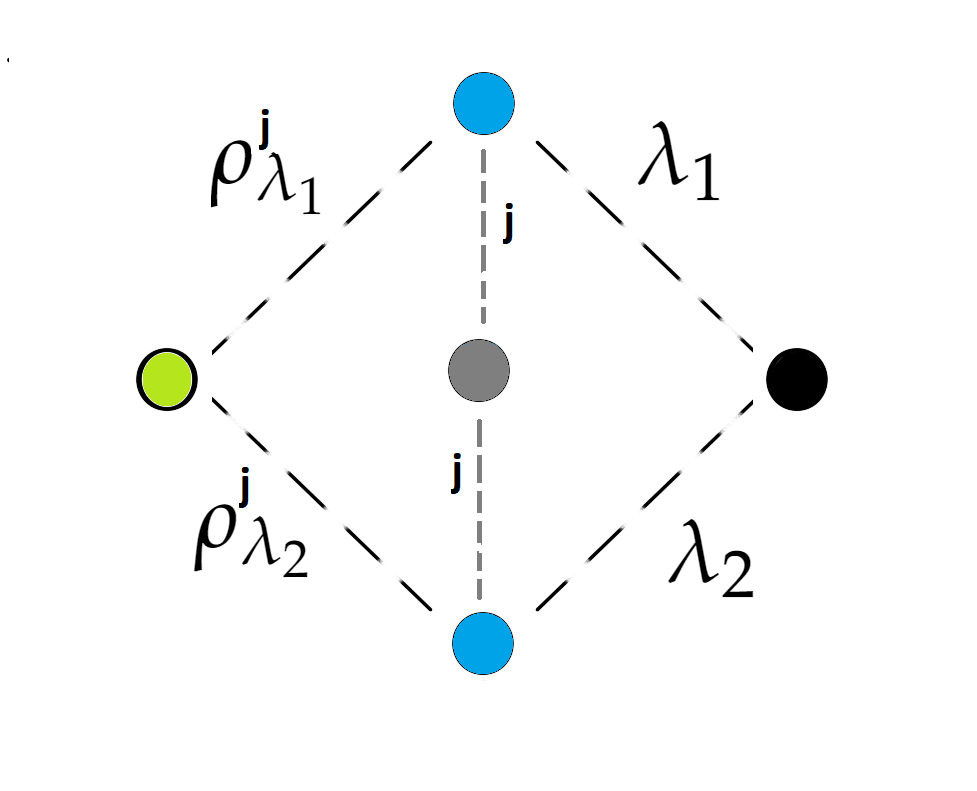} }}%
    \quad
{{\includegraphics[width=7cm]{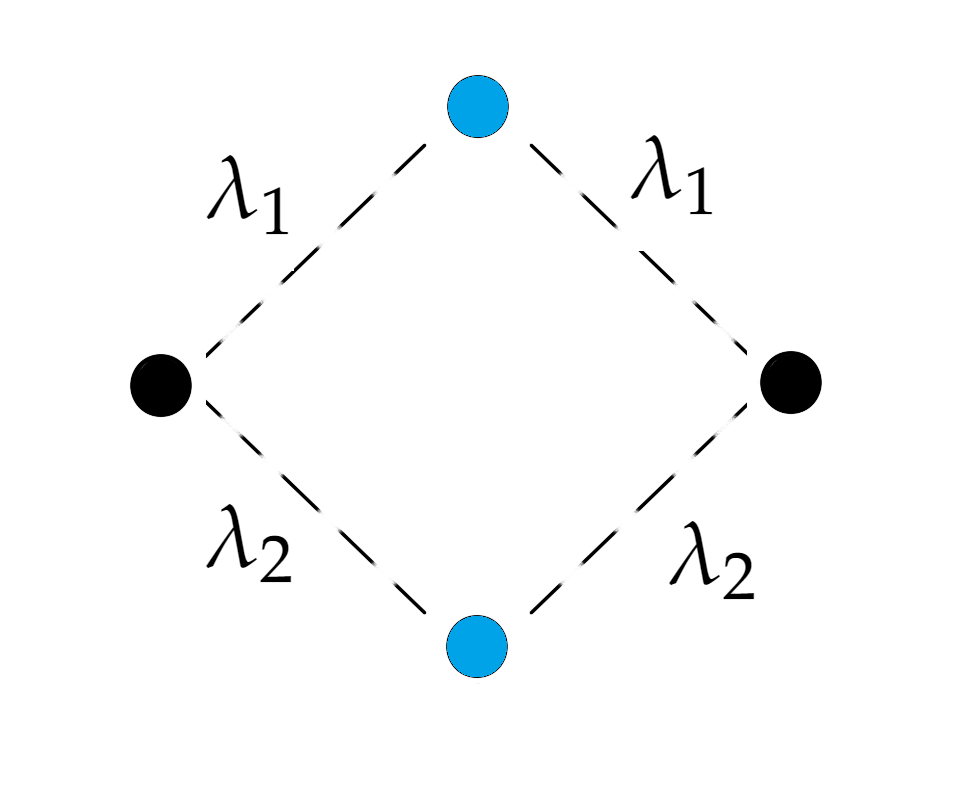} }}%
     \caption{Difference between SOHS and NLHV model in the minimal scenario. (left) Alice and Bob can explain the observed correlations $p(a,b)$ using a SOHS model. Alice is trusted and thus receives quantum states from the sources but there is no restriction over Bob. The grey box denotes an unknown source of classical random variables that might correlate the sources $S_1,S_2$. (right) Alice and Bob can explain the observed correlations $p(a,b)$ using a NLHV model.}
    \label{fig2}
\end{figure*}
Now, assuming separable quantum sources [assumption \ref{ass2}] we express $\rho_{\lambda_1,\lambda_2}$ using pure state decompositions to arrive at the following expression of $p(a,b)$
\begin{eqnarray}\label{SOHS2}
   p(a,b)=\sum_{\lambda_1,\lambda_2} p(\lambda_1,\lambda_2)\sum_{p^j_{\lambda_1,\lambda_2}}p^j_{\lambda_1,\lambda_2} p(a|\ \ket{\psi^j_{\lambda_1}}\ket{\psi^j_{\lambda_2}})p(b|\lambda_1,\lambda_2).
\end{eqnarray}
If correlations $\vec{p}$ admit the form \eqref{SOHS2}, then they are describable using a separable outcome-independent hidden state (SOHS) model. A simple example of the SOHS model would be that sources $S_1,S_2$ locally toss a coin, that is, $\lambda_{1/2}=\{1$(head),$2$(tail)\} based on which they send a state $\rho_{\lambda}$ to Alice and the outcome of the toss to Bob.

 
To witness swap-steering, a functional $W$ can be constructed which depends on $\vec{p}$ as
\begin{eqnarray}
W(\vec{p})=\sum_{a,b}c_{a,b}p(a,b) \leq\beta_{SOHS}
\end{eqnarray}
where $c_{a,b}$ are real coefficients and $\beta_{SOHS}$ denotes the maximum value attainable using assemblages admitting a SOHS model \eqref{SOHS2}. For the purpose of this article, we consider only functionals that are linear over $\vec{p}$. 

Now, consider the following functional 
\begin{equation}\label{steein}
W=p(0,0)+p(1,1)+p(2,2)+p(3,3)\leq \beta_{SOHS}
\end{equation}
Recall here that Alice is trusted and performs the measurements with elements given in \eqref{Amea1}. Let us now find the maximum value that can be achieved using correlations that admit a SOHS model \eqref{SOHS2}.

\begin{fakt}\label{fact1}
    Consider the swap-steering functional $W$ \eqref{steein}. The maximum value $\beta_{SOHS}$ that can be achieved using correlations that admit a SOHS model \eqref{SOHS2} of $W$ is $\beta_{SOHS}=\frac{1}{2}$.
\end{fakt}
\begin{proof} The proof follows the exact same lines as presented in \cite{sarkar6, sarkar11,sarkar12}. Let us  consider the steering functional $W$ in Eq. \eqref{steein} and express it in terms of the SOHS model \eqref{SOHS2} as
\begin{eqnarray}
    \sum_{a=0}^3\sum_{\lambda_1,\lambda_2} \ p(\lambda_1,\lambda_2)p(a|\rho_{\lambda_1,\lambda_2})p(a|\lambda_1,\lambda_2)
    \leq \sum_{\lambda_1,\lambda_2} \ p(\lambda_1,\lambda_2)\max_{a}\{p(a|\rho_{\lambda_1,\lambda_2})\}
\end{eqnarray}
where we used the fact that $\sum_ap(a|\lambda_1,\lambda_2)=1$ for any $\lambda_1,\lambda_2$. Now, maximising over $\rho_{\lambda_1,\lambda_2}$ gives us
\begin{eqnarray}
     \sum_{\lambda_1,\lambda_2} \ p(\lambda_1,\lambda_2)\max_{a}\{p(a|\rho_{\lambda_1,\lambda_2})\} \leq \sum_{\lambda_1,\lambda_2} \ p(\lambda_1,\lambda_2)\max_{\rho_{\lambda_1,\lambda_2}}\max_{a}\{p(a|\rho_{\lambda_1,\lambda_2})\}.\ \ \ 
\end{eqnarray}
Now, using the fact that $\sum_{\lambda_1,\lambda_2}p(\lambda_1,\lambda_2)=1$ for $i=1,2$ 
allows us to conclude that 
\begin{equation}
    \beta_{SOHS}\leq \max_{\ket{\psi}_{A_1},\ket{\psi}_{A_2}}\max_{a}\{p(a|\ \ket{\psi}_{A_1},\ket{\psi}_{A_2})\}.
\end{equation}
As the steering functional $W$ is linear, without loss of generality we consider the maximization only over pure states.  Now, putting in the measurement of the trusted Alice \eqref{Amea1}, which locally acts on qubit Hilbert spaces, and thus optimizing over pure states $\ket{\psi}_{A_1},\ket{\psi}_{A_2}\in \mathbbm{C}^2$ gives us $\beta_{SOHS}\leq\frac{1}{2}$. This bound can be saturated when the sources prepare the maximally mixed $\rho_i=\frac{1}{2}\left(\proj{00}+\proj{11}\right)_{A_iB_i}$ and the measurement with Bob is 
$\{\proj{00},\proj{01},\proj{10},\proj{11}\}$. This state clearly has a SOHS model and thus we get the desired SOHS bound.
\end{proof}

Consider that the sources prepare the state $\ket{\psi_i}=\ket{\phi_+}_{A_iB_i}$ and Bob performs the same measurement as Alice, that is, $M_B=\{\proj{\phi_{+}},\proj{\phi_{-}},\proj{\psi_{+}},\proj{\psi_{-}}\}_{B_0B_1}$ where the corresponding states are given in \eqref{Amea1}. Using these states and Bob's measurement one can simply evaluate the steering functional $W$ in \eqref{steein} to get the value $1$, which is the quantum bound of $W$. Notice that this is also the algebraic value of $W$.

Let us also show here that one can not observe Bell-type non-locality with only two parties without inputs. Without loss of generality, we consider here the scenario similar to one depicted in Fig. \ref{fig1} such that Alice and Bob perform a measurement with arbitrary number of outcomes on subsystems sent by two independent or classically correlated sources. However, unlike the previous scenario, Alice is untrusted. 
If the correlations $\vec{p}=\{p(a,b)\}$ admit a network-local hidden variable (NLHV) model \cite{renou1, supic4}, then they can be represented as
\begin{equation}\label{LHV1}
     p(a,b)=\sum_{\lambda_1,\lambda_2} p(\lambda_1)p(\lambda_2) p(a|\lambda_1,\lambda_2)p(b|\lambda_1,\lambda_2)
\end{equation}
for any $a,b$. Let us state the following fact which is simple to prove.

\begin{fakt}
Consider the scenario depicted in Fig. \ref{fig2}. The correlations $\vec{p}=\{p(a,b)\}$ obtained by Alice and Bob can always be described by an NLHV model \eqref{LHV1}.
\end{fakt}

\begin{proof} It is well-known that if Alice and Bob do not have inputs in the standard Bell scenario, then any joint correlation can be represented using an LHV model of the form
\begin{eqnarray}\label{LHV2}
    p(a,b)=\sum_{\lambda} p(\lambda)p(a|\lambda)p(b|\lambda).
\end{eqnarray}
Now, let us consider the scenario depicted in Fig. \ref{fig2} and consider that Alice and Bob's outcomes are independent of the source $S_2$
, that is, $p(a|\lambda_2)=p(a)$ and $p(b|\lambda_2)=p(b)$. 
Now, Eq. \eqref{LHV2} can be rewritten using $\lambda_2$ and the fact that $\sum_{\lambda_2}p(\lambda_2)=1$ as
\begin{eqnarray}
     p(a,b)=\sum_{\lambda,\lambda_2} p(\lambda)p(\lambda_2)p(a|\lambda,\lambda_2)p(b|\lambda,\lambda_2)
\end{eqnarray}
which is the form \eqref{LHV1}.
\end{proof}

The above fact can be straightforwardly generalized to the scenario with arbitrary number of sources between Alice and Bob. It is then well-known that one can not observe any non-locality without inputs when there is a single source distributing subsystems to Alice and Bob. Thus, to observe any form of quantum non-locality in the minimal possible scenario, in the sense that there are no inputs and only two parties, one has to trust either of the parties.  Consequently, quantum steering can also be observed in scenarios where one can not observe Bell non-locality. Let us now show that a class of states that is unsteerable in the standard quantum steering scenario is swap-steerable. 

\subsection{Entanglement assisted activation of steerability} Let us now consider the Werner state given by
\begin{eqnarray}\label{Werner}
    \rho^W(\alpha)=\alpha\proj{\phi_{+}}+(1-\alpha)\frac{\I}{4}.
\end{eqnarray}
The above state is separable iff $\alpha\leq\frac{1}{3}$ \cite{Werner1}. As proven in \cite{Wiseman, Bowles1}, the above  state is steerable in the standard quantum steering scenario iff $\alpha>\frac{1}{2}$. Thus, in the range of $\frac{1}{3}<\alpha\leq\frac{1}{2}$, the Werner state is unsteerable but entangled. We show here that the Werner state when coupled with the maximally entangled state is swap-steerable. Thus when assisted with entanglement, unsteerable states can be activated to display steerability without inputs.

\begin{fakt}
   The Werner state $\rho^W(\alpha)$ \eqref{Werner} with the maximally entangled state is swap-steerable for any $\alpha>\frac{1}{3}$.
\end{fakt}
\begin{proof}
    Consider the scenario presented in Fig. \ref{fig1}. Now, suppose that the source  $S_i$ generates the state $\rho^W_{A_iB_i}(\alpha_i)$ for $i=1,2$. Bob again performs the Bell basis measurement $M_B$. Given these states and measurements, let us again evaluate the steering functional $W$ in \eqref{steein} to obtain
    \begin{eqnarray}\label{Wenervalue1}
        W=\frac{3\alpha_1\alpha_2+1}{4}.
    \end{eqnarray}
    As proven above in Fact \ref{fact1}, if $W>\frac{1}{2}$ then the state is swap-steerable from Bob to Alice. Thus, we have from \eqref{Wenervalue1} that the Werner state \eqref{Werner} is steerable if $\frac{3\alpha_1\alpha_2+1}{4}>\frac{1}{2}$. Consequently, for any value of $\alpha_1\alpha_2>\frac{1}{3}$, the Werner states are swap-steerable. 
 Let us now observe that if $\alpha_1=1$, that is, the source $S_1$ generates maximally entangled state, then for any $\alpha_2>\frac{1}{3}$ the Werner state becomes swap-steerable. 
\end{proof}

Thus, some states that are unsteerable in the standard quantum steering scenario can be activated using the maximally entangled state and shown to be swap-steerable. However, we also notice that to observe swap-steering, the states generated by both sources can not be unsteerable simultaneously. 
Let us now find some necessary conditions to observe swap-steering.

\subsection{Necessary conditions for swap-steering} Consider again the scenario depicted in Fig. \ref{fig1}. Notice that one of the trivial necessary conditions to observe swap-steering is that the trusted party, here Alice, needs to perform an entangled measurement. Let us now restrict to the case when the number of outcomes on Bob's side is a composite number, that is, $b=b_0b_1$ where $b_0,b_1$ are positive integers. Now, Bob's measurement $\{N^b\}$ with $b=b_0b_1$ prepares a set of positive operators on the trusted Alice's side, known as assemblage, denoted as $\{\sigma_b\}$ where $\sigma_b=\sum_jp_j\Tr_{B}(\I_A\otimes N^b\rho_{A_1B_1}^j\otimes\rho_{A_2B_2}^j)$. Now, we show that if the assemblage is of a particular form, one can never observe swap-steering.

\begin{fakt}\label{fact4}
    Consider the swap-steering scenario depicted in Fig. \ref{fig1} where Alice and Bob share the states $\rho_{A_1B_1},\rho_{A_2B_2}$. Let us assume that Bob performs a $n-$outcome measurement which prepares the assemblage $\{\sigma_{b_0b_1}\}$ on the trusted Alice's side. If $\sigma_{b_0b_1}$ is separable for $b_0=0,1,\ldots,n_1-1,\ b_1=0,1,\ldots,n_2-1$,
    then there exists a SOHS model for both the states $\rho_{A_1B_1},\rho_{A_2B_2}$.
\end{fakt}
\begin{proof}
Let us first notice that 
\begin{eqnarray}
    \sum_{b_0,b_1}\sigma_{b_0b_1}&=&\sum_{b_0,b_1,j}p_j\Tr_{B}(\I_A\otimes N^b\rho_{A_1B_1}^j\otimes\rho^j_{A_2B_2})\nonumber\\&=&\sum_{j}p_j\rho_{A_1}^j\otimes\rho_{A_2}^j
\end{eqnarray}
which also allows us to conclude that $\sum_{b_0,b_1}\Tr(\sigma_{b_0b_1})
=1$. Consider now the assemblage $\{\sigma_{b_0b_1}\}$ is separable,  that is, 
the operators $\sigma_{b_0b_1}=\sum_j\sigma_{b_0}^j\otimes\sigma_{b_1}^j$. Notice that the following states  
\begin{eqnarray}\label{SOHS1}
    \tilde{\rho}_{A_iB_i}^j=\frac{1}{\mathcal{N}_{i,j}}\sum_{b_i=0}^{n_i-1}\sigma_{b_i,A_i}^j\otimes\proj{b_i}_{B_i}
\end{eqnarray}
where $\mathcal{N}_{i,j}=\sum_{b_i}\Tr(\sigma_{b_i}^j)$ and Bob performing a measurement of the form
\begin{eqnarray}
    \tilde{M}_{b_0b_1}=\proj{b_0}_{B_1}\otimes\proj{b_1}_{B_2} 
\end{eqnarray}
for $b_i=0,1\ldots,n_i-1$ 
give the same assemblage on Alice's side as the states $\sum_jp_j\rho_{A_1B_1}^j\rho^j_{A_2B_2}$ and the measurement $M_b=\{N^b\}$. It is straightforward to observe that the states $ \tilde{\rho}_{A_iB_i}$ are separable and thus the $\rho_{A_iB_i}$ admit a SOHS model.
\end{proof}
Consequently, one can observe from Fact \ref{fact4} that if Bob performs a product measurement, then the states are not swap-steerable from Bob to Alice. Further on, both states prepared from the sources are needed to be entangled to observe swap-steering. Thus, to observe swap-steering both the states and measurements must be entangled.

\section{Self-testing and randomness certification} Let us now utilise the above swap-steering inequality \eqref{steein} for self-testing the quantum realisations suggested after Fact \ref{fact1}. 
Self-testing in the 1SDI scenario was first defined 
in Ref. \cite{Supic, Alex}. Inspired by \cite{sarkar6, sarkar12, sarkar11}, we present a general definition of self-testing in the 1SDI scenario in quantum networks without inputs with one trusted party. Interestingly, we do not require assuming a pure underlying state or projective measurements. For a note, we express the measurements of both parties in the observable picture and represent it as $A_0,B_0$. For a discussion on observables refer to Appendix A.

Let us revisit the previous experiment in which Alice and Bob conduct measurements on the states $\rho_{AB}$ prepared by the sources $S_i\ (i=1,2)$ and observe the correlations ${p(a,b)}$. It is important to note that Alice's observables $A_0$ is fixed, whereas Bob's observables $B_0$ is arbitrary. Now, let us examine a reference experiment that reproduces the same statistics as the actual experiment but involves the states $\tilde{\rho}_{AB}$ and observables represented by $\tilde{B}_0$, which both parties wish to validate. The states $\rho_{AB}$ and the observables $B_0$ are self-tested from $\{p(a,b)\}$ if there exists a unitary $U_B:\mathcal{H}_B\to \mathcal{H}_B$ such that 
\begin{equation}
 (\mathbbm{1}_A\otimes U_B)\rho_{AB}(\mathbbm{1}_A\otimes U_B^{\dagger})=\tilde{\rho}_{AB'}\otimes\rho_{B''},
\end{equation}
\begin{equation}
    U_B\,B_0\,U_B^{\dagger}=\tilde{B}_0\otimes\mathbbm{1}_{B''},
\end{equation}
where $\mathcal{H}_B$ decomposes as $\mathcal{H}_B=\mathcal{H}_{B'}\otimes\mathcal{H}_{B''}$ such that $\mathcal{H}_{B''}$ denotes the junk Hilbert space. The states $\rho_{B''}$ and $\mathbbm{1}_{B''}$ denote the junk state and the identity acting on $\mathcal{H}_{B''}$.

Let us now state our self-testing statement but before proceeding, let us define  Alice's observable corresponding to the Bell basis as
\begin{eqnarray}\label{A0}
    A_0=\sum_{k=1}^4\mathbbm{i}^{k}\proj{\phi_k}.
\end{eqnarray}
where $\ket{\phi_1}=\ket{\phi^+},\ket{\phi_2}=\ket{\psi^+},\ket{\phi_3}=\ket{\phi^-}, \ket{\phi_4}=\ket{\psi^-}$.

\begin{fakt}\label{Theo1} 
Assume that the steering inequality \eqref{steein},  with trusted Alice choosing the observable $A_0$ \eqref{A0}, is maximally violated by a separable state $\rho_{AB}$ acting on $\mathbbm{C}^2\otimes\mathbbm{C}^2\otimes\mathcal{H}_B$ and Bob's observable $B_0$. Then, the following statements hold true:
\\
\\
1. Bob's measurement is projective with his Hilbert space decomposing as $\mathcal{H}_{B}=(\mathbbm{C}^2)_{B_1'}\otimes(\mathbbm{C}^2)_{B_2'}\otimes \mathcal{H}_{B_{12}''}$ for some auxiliary Hilbert space $\mathcal{H}_{B''_{12}}=\mathcal{H}_{B''_{1}}\otimes \mathcal{H}_{B''_{2}}$.\\
\\
2.  \ \  There exist unitary transformations, $U_{i}:\mathcal{H}_B\rightarrow\mathcal{H}_B$,  such that
\begin{eqnarray}\label{lem1.2}
(\mathbbm{1}_{A}\otimes U_B)\rho_{AB}(\mathbbm{1}_{A}\otimes U_B^{\dagger})=\proj{\phi^+}_{A_1B_1'}\otimes\proj{\phi^+}_{A_2B_2'}\otimes \rho_{B_1''B_2''},\ \ 
\end{eqnarray}
where $B_i''$ denotes Bob's auxiliary system, and 
\begin{eqnarray}\label{lem1.1}
\quad U_B\,B_0\,U_B^{\dagger}=A_0\otimes \mathbbm{1}_{B_1''B_2''}
\end{eqnarray}
where $U_B=U_1\otimes U_2$.
\end{fakt}

The proof of the above fact is given in Appendix \ref{A}. An interesting application of the above self-testing statement is that the untrusted Bob's measurement device can generate true randomness that is secure against adversaries. For this purpose, we consider an eavesdropper, Eve, who cannot directly read Bob's outcomes but may have correlations with him that she can exploit to infer his results. Consequently, we consider a state $\rho_{ABE}$ which is shared among Alice, Bob and Eve. As Eve's dimension is unrestricted, we can purify the state as $\ket{\psi_{ABE}}$ such that $\Tr_E\psi_{ABE}=\rho_{AB}$ where $\rho_{AB}$ is separable. 

Now, to certify whether the measurement outcomes as observed by Bob is truly random, we consider that Eve wants to guess the outcome of Bob's measurement. In order to do so,  she performs a measurement $Z=\{E_e\}$ on her part of the shared states. Here the outcome $e$ is Eve's best guess of Bob's outcome. However, any operation by Eve should not alter the statistics $\vec{p}=\{p(a,b)\}$ observed by Alice and Bob, that is,
\begin{eqnarray}
    p(a,b)=\bra{\psi}M_a\otimes N_b\otimes\I_E\ket{\psi}.
\end{eqnarray}
This is extremely important as the adversary Eve would like to remain invisible to Alice and Bob.

The number of random bits that can be securely generated from Bob's measurement is quantified as $H_{\min}=-\log_2 G(y,\vec{p})$ \cite{di4}, where $G(y,\vec{p})$ is known as the local guessing probability which can be computed as,
\begin{equation}\label{LGpr}
    G(\vec{p})=\sup_{S\in S_{\vec{p}}}\sum_{b}\bra{\psi} \I_A \otimes N_{b} \otimes E_b\ket{\psi},
\end{equation}
where $S_{\vec{p}}$ is the set of all Eve's strategies comprising of the shared states and her measurement that reproduce the probability distribution $\vec{p}$ as expected by Alice and Bob.

Let us now suppose that the swap-steering inequality \eqref{steein} is maximally violated
by $\vec{p}$. As proven above in Fact \ref{Theo1}, this implies that the state shared by Alice, Bob, and Eve up to local unitary operations is, $\ket{\psi_{{ABE}}}=\ket{\phi^+_{A_1B_1'}}\ket{\phi^+_{A_2B_2'}}\ket{\mathrm{\mathrm{aux}}_{\mathrm{B_{12}''E}}}$ as well as $N_b=\proj{\phi_b}\otimes\I_{\mathrm{B''_{12}}}$ where $\ket{\phi_b}$ are given above Eq. \eqref{A0}. Putting these states and measurement in the formula \eqref{LGpr} we obtain 
\begin{eqnarray}\label{guess1}
    G(\vec{p})=\sum_{b}\bra{\phi^+}\bra{\phi^+} (\I_A \otimes \proj{\phi_b})\ket{\phi^+}\ket{\phi^+}  \bra{\mathrm{aux}}\I_{B_{12''}}\otimes E_b\ket{\mathrm{aux}}.\qquad
\end{eqnarray}
Now for all $b$, $\bra{\phi^+}\bra{\phi^+} (\I_A \otimes \proj{\phi_b})\ket{\phi^+}\ket{\phi^+}=1/4$ which allows us to conclude from \eqref{guess1} that
\begin{eqnarray}
     G(\vec{p})=\frac{1}{4}\sum_{b}\bra{\mathrm{aux}}\I_{B_{12''}}\otimes E_b\ket{\mathrm{aux}}=\frac{1}{4}.
\end{eqnarray}
Consequently, $-\log_2 G(\vec{p})=2$ bits of randomness can be certified from Bob's measurement outcomes using our self-testing scheme.

It is important to note here that the generation of secure randomness is based on the assumption that the sources can only be correlated in a classical way. However, the adversary can always guess the outcomes of Bob if she manages to entangle the sources. For instance, (i) she can prepare both devices beforehand or (ii) she herself could perform an entangled measurement on the systems arriving on Bob's side and then send the outcome to Bob. This problem would persist in any security protocols involving two different constrained sources. However, the second type of attack (ii) can be avoided if Bob randomly chooses not to perform a measurement in some runs of the experiment. Since Eve is unaware of this fact, she would still entangle both sources and can be detected by Alice. 
It will be extremely interesting if Alice and Bob can perform some local operations on their subsystems to figure out whether the received subsystems are generated from separable sources or not.

\section{Discussions} 
The idea of quantum steering in networks was introduced recently in \cite{netstee}. However, the scenario considered in this work was not dealt with in Ref. \cite{netstee}. Further on, the notion of quantum steering in networks \cite{netstee} required the trusted party to perform a full tomography which implied that the trusted party has inputs. Contrary to this, in the swap-steering scenario described above even the trusted party performs a single fixed measurement. This also makes our scheme experimentally friendly as one has to consider less number of correlations in order to witness quantum steering in networks. However, the measurement elements of the trusted party are maximally entangled and thus it would be beneficial to explore the possibilities of observing swap-steering with less entangled measurements. 

Constructing witnesses to observe quantum nonlocality in networks has been extremely difficult mainly due to the fact that the network-local polytope might not be convex as shown in \cite{pironio1} [see nevertheless Ref. \cite{Sarkar_2024}]. In this work, we find that assuming one of the parties to be trusted allows constructing linear witnesses to observe a form of quantum nonlocality in networks. One of the interesting follow-up directions would be to explore the structure of the set of correlations admitting the SOHS model. 
We showed in this work that any entangled Werner state can be used to witness swap-steering. An interesting follow-up question is whether every entangled state violates the notion of swap-steering. This problem has now been resolved for every bipartite entangled state in \cite{sarkar2024witnessing}. Another direction to explore will be toward generalizing the notion of swap-steering to more parties and outcomes. It is known that quantum steering is asymmetric, that is, there are quantum states that are steerable from Alice to Bob but not the other way around. It will be interesting to find similar properties of quantum states when considering the notion of swap-steering. 
Furthermore, we used swap-steering for the certification of randomness without seed randomness. It will be highly desirable to generalize the above scheme to the DI regime where no party is trusted. Another direction would be to generalize the scheme presented in this work to certify an unbounded amount of randomness. Moreover, it would be interesting to investigate whether the randomness certification can be made robust to experimental imperfections.

\begin{acknowledgments}
 We would like to thank Stefano Pironio for reviewing the manuscript and providing critical comments that considerably improved the manuscript. This project was funded within the QuantERA II Programme (VERIqTAS project) that has received funding from the European Union’s Horizon 2020 research and innovation programme under Grant Agreement No 101017733.
\end{acknowledgments}

\printbibliography

\section{Appendix}
\appendix

\section{Self-testing}\label{A}
\setcounter{fakt}{4}

In quantum theory, it is advantageous to express the correlations $\{p(a,b)\}$in terms of expectation values rather than probability distributions. When dealing with $d$-outcome measurements, a useful technique is to utilize the two-dimensional Fourier transform of the conditional probabilities $p(a,b)$ as
\begin{equation}\label{ExpValues}
    \langle A^{(k)}_0 B^{(l)}_0 \rangle = \sum^{d-1}_{a,b=0} \w^{ak+bl} p(a,b),
\end{equation}
where $\w$ is the $d$-th root of unity $\omega=\exp(2\pi\mathbbm{i}/d)$ and $k,l=0,\ldots,d-1$ and $ A^{(k)}_0,B^{(l)}_0$ are known as observables. 
Using the inverse Fourier transform of \eqref{ExpValues}, we obtain that
\begin{eqnarray}\label{ExpValues1}
    p(a,b)=\frac{1}{d^2}\sum_{k,l=0}^{d-1}\omega^{-(ak+bl)} \langle A^{(k)}_0 B^{(l)}_0 \rangle.
\end{eqnarray}

The expectation value appearing on the left-hand side of Eq. (\ref{ExpValues}) can be simply represented as $\langle A^{(k)}_0 B^{(l)}_0 \rangle = \Tr(A^{(k)}_0 \otimes B^{(l)}_0\rho_{AB})$ for some state $\rho_{AB}$ with $\{A^{(k)}_0\}$ and $\{B^{(l)}_0\}$ are operators defined as
\begin{equation}\label{obsgen}
    A^{(k)}_0= \sum^{d-1}_{a=0} \w^{ak} P^{(a)}, \qquad B^{(l)}_0 = \sum^{d-1}_{b=0} \w^{bl} Q^{(b)}.
\end{equation}
where $P^{(a)},Q^{(b)}$ represent the measurement elements of Alice, Bob respectively.
As proven in \cite{Jed1}, the observables $ A^{(k)}_0$ have the following properties (same for $B^{(l)}_0$): $A^{(d-k)}_0=(A^{(k)}_0)^{\dagger}$ and $A^{(k)}_0(A^{(k)}_0)^{\dagger}\leq\I$. For the special case of projective measurements, the observables $A^{(k)}_0$ are unitary and $ A^{(k)}_0= (A^{(1)}_0)^k=A^k_0$.
As Alice performs the Bell-basis measurement whose corresponding measurement elements for the rest of the manuscript will be denoted as $\ket{\phi_1}=\ket{\phi^+},\ket{\phi_2}=\ket{\psi^+},\ket{\phi_3}=\ket{\phi^-}, \ket{\phi_4}=\ket{\psi^-}$ and the corresponding observable using \eqref{obsgen} is given as
\begin{eqnarray}\label{A0}
    A_0=\sum_{k=1}^4\mathbbm{i}^{k}\proj{\phi_k}.
\end{eqnarray}

Let us first revisit the swap-steering inequality \eqref{steein} and then using \eqref{ExpValues1}, the above steering inequality can be simply represented as
\begin{eqnarray}\label{steein1}
  W= \frac{1}{4}\sum_{k=0}^3 \langle A_0^k\otimes B_0^{(4-k)} \rangle\leq \beta_{LHS}
\end{eqnarray}
The quantum bound of the above steering inequality is $1$ which is also the maximum algebraic value of $W$. Consequently, we observe from \eqref{steein1} that the maximum value can be attained iff each term is $1$, that is, for $k=0,1,2,3$
\begin{eqnarray}
    \langle A_0^k\otimes B_0^{(4-k)} \rangle=1.
\end{eqnarray}
Now, using Cauchy-Schwarz inequality we get that
\begin{eqnarray}\label{SOS3}
     A_0^k\otimes B_0^{(4-k)}\rho_{AB}=\rho_{AB}.
\end{eqnarray}
Recalling that $\rho_{AB}$ is separable, we can express it as $\rho_{AB}=\sum_jp_j\ \rho^j_{A_1B_1}\otimes\rho^j_{A_2B_2}$ which using its eigendecomposition can be expressed as $\rho_{AB}=\sum_{s,s'}p_{s,s'}\proj{\psi_{s,A_1B_1}}\otimes\proj{\psi_{s',A_2B_2}}$. Consequently, we get from the above expression Eq. \eqref{SOS2} that
\begin{equation}\label{SOS1}
    \sum_{s,s'}p_{s,s'} A_0^k\otimes B_0^{(4-k)}\ \psi^1_s\otimes\psi_{s'}^2=\sum_{s,s'}p_{s,s'}\ \psi^1_s\otimes\psi_{s'}^2
\end{equation}
where for simplicity, we represent the states $\proj{\psi_{s,A_iB_i}}$ as $\psi^i_s$.
It is now straightforward to observe from the above relation that for all $s,s'$
\begin{eqnarray}\label{SOS2}
     A_0^k\otimes \overline{B}_{0,ss'}^{4-k}\ \ket{\psi^1_s}\ket{\psi^2_{s'}}=\ket{\psi^1_s}\ket{\psi^2_{s'}}
\end{eqnarray}
Here $\overline{B}_{0,ss'}$ is the projection of $B_0$ on the support of $\Tr_A\psi^1_s\otimes\Tr_A\psi^2_{s'}$.
The above relations are sufficient to self-test the state $\rho_{AB}$ and Bob's measurement $B_0$. Before proceeding toward the self-testing result, it is important to recall the assumption that the local states are full-rank as the measurements can only be characterized on the local support of the states.
For a note, we closely follow the techniques introduced in \cite{sarkar6}. 

\setcounter{fakt}{4}
\begin{fakt}\label{Theo1M} 
Assume that the steering inequality \eqref{steein1},  with trusted Alice choosing the observable $A_0$ \eqref{A0}, is maximally violated by a separable state $\rho_{AB}$ acting on $\mathbbm{C}^2\otimes\mathbbm{C}^2\otimes\mathcal{H}_B$ and Bob's observable $B_0$. Then, the following statements hold true:
\\
\\
1. Bob's measurement is projective with his Hilbert space decomposing as $\mathcal{H}_{B}=(\mathbbm{C}^2)_{B_1'}\otimes(\mathbbm{C}^2)_{B_2'}\otimes \mathcal{H}_{B_{12}''}$ for some auxiliary Hilbert space $\mathcal{H}_{B''_{12}}=\mathcal{H}_{B''_{1}}\otimes \mathcal{H}_{B''_{2}}$.\\
\\
2.  \ \  There exist unitary transformations, $U_{i}:\mathcal{H}_B\rightarrow\mathcal{H}_B$,  such that
\begin{eqnarray}\label{lem1.2}
(\mathbbm{1}_{A}\otimes U_B)\rho_{AB}(\mathbbm{1}_{A}\otimes U_B^{\dagger})\qquad\qquad\qquad\qquad\nonumber\\=\proj{\phi^+}_{A_1B_1'}\otimes\proj{\phi^+}_{A_2B_2'}\otimes \rho_{B_1''B_2''},
\end{eqnarray}
where $B_i''$ denotes Bob's auxiliary system, and 
\begin{eqnarray}\label{lem1.1}
\quad U_B\,B_0\,U_B^{\dagger}=A_0\otimes \mathbbm{1}_{B_1''B_2''}
\end{eqnarray}
where $U_B=U_1\otimes U_2$.
\end{fakt}
\begin{proof}
Let us first show that Bob's measurement is projective. For this purpose, we consider the relations \eqref{SOS3} for $k=1$ and then multiply it with $A_0^3\otimes B_0$ to obtain
\begin{eqnarray}\label{Bobmea1}
    \I_A\otimes B_0B_0^{(3)}\ \rho_{AB}=A_0^3\otimes B_0\ \rho_{AB}
\end{eqnarray}
where we used the fact that $A_0^4=\I_A$. 
Notice that the right-hand side of the above expression \eqref{Bobmea1} can be simplified using the relation \eqref{SOS3} for $k=3$ to obtain
\begin{eqnarray}
     \I_A\otimes B_0B_0^{(3)}\ \rho_{AB}= \rho_{AB}.
\end{eqnarray}
Thus, taking a partial trace over Alice's subsystem and recalling that $B_0^{(3)}=B_0^{\dagger}$ gives us
\begin{eqnarray}
B_0B_0^{\dagger}\ \rho_B=\rho_B
\end{eqnarray}
where $\rho_B=\Tr_B\  \rho_{AB}$. As the local states are full-rank, they are invertible too and consequently one can arrive at
\begin{eqnarray}
    B_0B_0^{\dagger}=\I_B.
\end{eqnarray}
Similarly, one can also find that $B_0^{\dagger}B_0=\I_B$. Both these relations of Bob's observable suggest that the observable $B_0$ and unitary, and thus Bob's measurement is projective. In a similar manner, considering the relation \eqref{SOS2} one can observe that  $\overline{B}_{0,ss'}$ for all $s,s'$ are unitary.

Let us now consider the relation Eq. \eqref{SOS2} and characterize the states $\ket{\psi^1_s},\ket{\psi^2_{s'}}$ that satisfy the relation \eqref{SOS2}. For simplicity, we drop the indices $s,s'$ for now.
As the local states on Alice's side belong to $\mathbb{C}^2$, using Schmidt decomposition we represent  $\ket{\psi^1},\ket{\psi^2}$ as
\begin{eqnarray}
\ket{\psi^i}=\sum_{j=0,1}\lambda_{j,i}\ket{e_{j,i}}\ket{f_{j,i}}
\end{eqnarray}
where $\lambda_{j,i}\geq0$ and $\{\ket{e_{j,i}}\},\{\ket{f_{j,i}}\}$ form an orthonormal basis for each $i$.
Now applying a unitary $U_{i}$ on these states such that $U_{i}\ket{f_{j,i}}=\ket{e^*_{j,i}}$ gives us
\begin{eqnarray}
\ket{\tilde{\psi}^i}=U_{i}\ket{\psi^i}=\sum_{j=0,1}\lambda_{j,i}\ket{e_{j,i}}\ket{e^*_{j,i}}.
\end{eqnarray}
Now, notice that the state on the right-hand side can be represented as
\begin{eqnarray}\label{state1}
\ket{\tilde{\psi}^i}= P_{i}\otimes\I_{B_i}\ket{\phi^+}
\end{eqnarray}
where
\begin{eqnarray}\label{P}
    P_i=\sqrt{2}\sum_{j=0,1}\lambda_{j,i}\proj{e_{j,i}}.
\end{eqnarray}
Notice that $P_i$ is full-rank as states that are separable between Alice and Bob can not violate the swap-steering inequality \eqref{steein}.
Putting the state \eqref{state1} in the relation \eqref{SOS2} gives us
\begin{eqnarray}\label{st22}
    A_0(P_1\otimes P_2) \otimes \tilde{B}_0^{\dagger}\ket{\phi^+}\ket{\phi^+}=P_1\otimes P_2 \ket{\phi^+}\ket{\phi^+}
\end{eqnarray}
where $\tilde{B}_0=U_{1}^{\dagger}\otimes U_2^{\dagger}\ \overline{B}_0\ U_{1}\otimes U_2$.
Now, using the fact that
\begin{eqnarray}
\ket{\phi^+}_{A_1B_1}\ket{\phi^+}_{A_2B_2}=\ket{\phi^+_4}_{A_1A_2|B_1B_2}
\end{eqnarray}
where $\ket{\phi^+_4}$ is the maximally entangled state of local dimension four. This allows us to conclude from \eqref{st22} that
\begin{eqnarray}
   (P_1^{-1}\otimes P_2^{-1}) A_0(P_1\otimes P_2) \otimes \tilde{B}_0^{\dagger}\ \ket{\phi^+_4}=\ket{\phi^+_4}.
\end{eqnarray}
Now, using the fact that $R\otimes Q\ket{\phi^+}=RQ^T\otimes \I\ket{\phi^+}$, where $T$ denotes the transpose in the computational basis, gives us
\begin{eqnarray}
     (P_1^{-1}\otimes P_2^{-1}) A_0(P_1\otimes P_2) \tilde{B}_0^{*}\otimes\I_B\ket{\phi^+_4}=\ket{\phi^+_4}.
\end{eqnarray}
Taking the partial trace over $B's$ subsystem allows us to conclude that
\begin{eqnarray}
     (P_1^{-1}\otimes P_2^{-1}) A_0(P_1\otimes P_2) \tilde{B}_0^{*}=\I_A
\end{eqnarray}
which eventually leads us to Bob's measurement being
\begin{eqnarray}\label{Bmea1}
    \tilde{B}^T_0= (P_1^{-1}\otimes P_2^{-1}) A_0(P_1\otimes P_2).
\end{eqnarray}
As $\tilde{B}_0$ is unitary and $P_1,P_2$ are Hermitian, we get from the above condition that
\begin{eqnarray}
     (P_1^{-1}\otimes P_2^{-1}) A_0(P_1\otimes P_2)^2A_0^{\dagger}(P_1^{-1}\otimes P_2^{-1})=\I_A.
\end{eqnarray}
Rearranging the terms we obtain that
\begin{eqnarray}
     A_0(P_1\otimes P_2)^2=(P_1\otimes P_2)^2A_0
\end{eqnarray}
which is equivalent to 
\begin{eqnarray}
    [A_0,(P_1\otimes P_2)^2]=0.
\end{eqnarray}
Now, notice that if two matrices commute then they share the same basis. However, the matrix $A_0$ has an entangled basis and the matrix $P_1\otimes P_2$ have a product basis. Thus, the only instance for these two matrices to commute is when $P_1\otimes P_2=\I$ which imposes that $P_1=P_2=\I$. Going back to Eq. \eqref{state1} allows us to conclude that the states $\ket{\psi^1},\ket{\psi^2}$ are the maximally entangled state, that is,
\begin{eqnarray}
   \I_A\otimes U_{i}\ket{\psi^i}= \ket{\phi^+}\quad i=1,2
\end{eqnarray}
and Bob's measurement using \eqref{Bmea1} is
\begin{eqnarray}
  U_{1}^{\dagger}\otimes U_2^{\dagger}\ \overline{B}_0\ U_{1}\otimes U_2=A_0^T=A_0.
\end{eqnarray}

Let us now bring back the indices $s,s'$ and rewrite the states and measurements as
\begin{eqnarray}
\ket{\psi^i_{s}}=\frac{1}{\sqrt{2}}\sum_{j=0,1}\ket{j}\ket{f_{j,i,s}}
\end{eqnarray}
where $U_{s,i}^{\dagger}\ket{j}=\ket{f_{j,i,s}}$ and 
\begin{eqnarray}
     \overline{B}_{0,ss'}=U_{s,1}\otimes U_{s',2}\ A_0\ U_{s,1}^{\dagger}\otimes U_{s',2}^{\dagger}
\end{eqnarray}
for all $s,s'$. From Theorem 1.1 of \cite{sarkar6}, we can express $B_0$ as 
\begin{eqnarray}
 B_0=\overline{B}_{0,ss'}\oplus E_{ss'}
\end{eqnarray}
where $E_{ss'}$ are unitary matrices.

Let us now denote Bob's local support of the states $\ket{\psi^i_s}$ as $V_{i,s}=\ $span$\{\proj{f_{0,i,s}},\proj{f_{1,i,s}}\}$ for all $i,s$. Further on, we will show that the supports $V_{i,l}, V_{i,l'}$ are orthogonal for any $l,l'$. For this purpose, we first express the product of the states $\ket{\psi^1_s}\ket{\psi^2_{s'}}$ as
\begin{eqnarray}
\ket{\psi^1_s}\ket{\psi^2_{s'}}=\frac{1}{2}\sum_{i,j=0,1}\ket{ij}\ket{f_{i,1,s}}\ket{f_{j,2,s'}}
\end{eqnarray}
which can equivalently be expressed using the Bell basis as
\begin{eqnarray}\label{STATE3}
\ket{\psi^1_s}\ket{\psi^2_{s'}}=\frac{1}{2}\sum_{i=1}^4\ket{\phi_i}\ket{g^i_{ss'}}
\end{eqnarray}
where $\ket{\phi_i}$ are given just above Eq. \eqref{A0} and 
\begin{eqnarray}\label{g1}
    \ket{g^1_{ss'}}=\frac{1}{\sqrt{2}}\left(\ket{f_{0,1,s}}\ket{f_{0,2,s'}}+\ket{f_{1,1,s}}\ket{f_{1,2,s'}}\right)\nonumber\\
    \ket{g^2_{ss'}}=\frac{1}{\sqrt{2}}\left(\ket{f_{0,1,s}}\ket{f_{1,2,s'}}+\ket{f_{1,1,s}}\ket{f_{0,2,s'}}\right)\nonumber\\
    \ket{g^3_{ss'}}=\frac{1}{\sqrt{2}}\left(\ket{f_{0,1,s}}\ket{f_{0,2,s'}}-\ket{f_{1,1,s}}\ket{f_{1,2,s'}}\right)\nonumber\\
    \ket{g^4_{ss'}}=\frac{1}{\sqrt{2}}\left(\ket{f_{0,1,s}}\ket{f_{1,2,s'}}-\ket{f_{1,1,s}}\ket{f_{0,2,s'}}\right)
\end{eqnarray}
Let us again utilize the relation \eqref{SOS2} and apply the state \eqref{STATE3} to it to observe that
\begin{eqnarray}
\sum_{i=1}^4\omega^i\ket{\phi_i}B_0^3\ket{g^i_{ss'}}=\sum_{i=1}^4\ket{\phi_i}\ket{g^i_{ss'}}.
\end{eqnarray}
Multiplying with $\bra{\phi_i}$ on both sides of the above expression gives us
\begin{eqnarray}\label{SOS4}
\omega^iB_0^3\ket{g^i_{ss'}}=\ket{g^i_{ss'}}\qquad \forall i.
\end{eqnarray}
As $B_0$ is unitary, we can conclude from the above formula \eqref{SOS4} that
\begin{eqnarray}\label{SOS5}
   \langle g^j_{ll'} \ket{g^i_{ss'}}=0 \qquad i\ne j
\end{eqnarray}
for any $i,j,l,l',s,s'$. Let us now consider Eq. \eqref{SOS5} with $l=s, j=1$ and expand it using \eqref{g1} to obtain the following conditions for $i=2,3,4$ as
\begin{subequations}
\begin{equation}\label{abcd1}
\langle{f_{0,2,l'}}\ket{f_{0,2,s'}}-\langle{f_{1,2,l'}}\ket{f_{1,2,s'}}=0
\end{equation}
\begin{equation}\label{abcd2}
\langle{f_{0,2,l'}}\ket{f_{1,2,s'}}+\langle{f_{1,2,l'}}\ket{f_{0,2,s'}}=0
\end{equation}
\begin{equation}\label{abcd3}
\langle{f_{0,2,l'}}\ket{f_{1,2,s'}}-\langle{f_{1,2,l'}}\ket{f_{0,2,s'}}=0.
\end{equation}
\end{subequations}
From Eqs. \eqref{abcd2} and \eqref{abcd3}, it is straightforward to observe that $\langle{f_{0,2,l'}}\ket{f_{1,2,s'}}=\langle{f_{1,2,l'}}\ket{f_{0,2,s'}}=0$. Let us now recall that $\ket{\psi^2_{l'}}$ and $\ket{\psi^2_{s'}}$ are orthogonal as they correspond to two different eigenvectors of $\rho_{AB}$ which gives us an additional condition
\begin{eqnarray}\label{abcd4}
\langle{f_{0,2,l'}}\ket{f_{0,2,s'}}+\langle{f_{1,2,l'}}\ket{f_{1,2,s'}}=0.
\end{eqnarray}
It is again straightforward to observe from \eqref{abcd1} and \eqref{abcd4} that $\langle{f_{0,2,l'}}\ket{f_{0,2,s'}}=\langle{f_{1,2,l'}}\ket{f_{1,2,s'}}=0$. Thus, the local supports $V_{2,s'}$ and $V_{2,l'}$ are orthogonal for any $s',l'$ such that $s'\ne l'$. Proceeding the same way as above, we can also conclude that the local supports $V_{1,s}$ and $V_{1,l}$ are orthogonal for any $s,l$ such that $s\ne l$. Consequently, the local supports $V_{ss'}=V_{1,s}\otimes V_{2,s'}$ are mutually orthogonal for any $s,s'$.

The local supports $V_{ss'}$ being mutually orthogonal imply that Bob's Hilbert space admits the following decomposition
\begin{equation}\label{block1}
    \mathcal{H}_B= \bigoplus_{s}\bigoplus_{s'}V_{ss'}=\bigoplus_{s}V_{1,s}\otimes\bigoplus_{s'}V_{2,s'}.
\end{equation}
As $\dim V_{1,s}=\dim V_{2,s'}=2$ for any $s,s'$, we can straightforwardly conclude that $\mathcal{H}_B=(\mathbbm{C}^2)_{B'_1}\otimes(\mathbbm{C}^2)_{B'_1}\otimes\mathcal{H}_{B''_{12}}$ where $\mathcal{H}_{B''_{1}}\otimes \mathcal{H}_{B''_{2}}$ for some Hilbert spaces
$\mathcal{H}_{B''_i}$.

The rest of the proof is exactly the same as step 3 in Theorem 1.2 of \cite{sarkar6}, which allows us conclude that there exist unitary transformations, $U_{i}:\mathcal{H}_B\rightarrow\mathcal{H}_B$,  such that
\begin{eqnarray}
(\mathbbm{1}_{A}\otimes U_1\otimes U_2)\rho_{AB}(\mathbbm{1}_{A}\otimes U_1^{\dagger}\otimes U_2^{\dagger})=\proj{\phi^+}_{A_1B_1'}\otimes\proj{\phi^+}_{A_2B_2'}\otimes \rho_{B_1''B_2''},
\end{eqnarray}
where $\rho_{B_1''B_2''}$ denotes Bob's auxiliary state which is separable with 
\begin{eqnarray}
    U_i=\bigoplus_sU_{s,i} \qquad i=1,2
\end{eqnarray}
and
\begin{eqnarray}
U_1\otimes U_2 \,B_0\,U_1^{\dagger}\otimes U_2^{\dagger}=A_0\otimes \mathbbm{1}_{B_1''B_2''}.
\end{eqnarray}
This completes the proof.
\end{proof}

\end{document}